\documentclass[10pt]{amsart}
\usepackage{graphicx}
\usepackage{amssymb}
\usepackage{amsfonts}
\usepackage{amsmath}
\usepackage{amsthm}
\usepackage[all]{xy}
\newtheorem{theorem}{Theorem}[section]
\newtheorem{lemma}[theorem]{Lemma}

\theoremstyle{definition}

\numberwithin{equation}{section}
 \theoremstyle{plain}    
 
 \theoremstyle{remark}    
 \newtheorem*{acknowledgement*}{Acknowledgement} 
\newcommand{\cF}{\mathcal{F}}

\newcommand{\cP}{\mathcal{P}}

\newcommand{\cT}{\mathcal{T}}

\newcommand{\cX}{\mathcal{X}}

\newcommand{\te}{{\theta}}
\newcommand{\Te}{{\Theta}}

\newcommand{\Om}{{\Omega}}
\newcommand{\om}{{\omega}}

\newcommand{\del}{{\delta}}
\newcommand{\Del}{{\Delta}}
\newcommand{\gam}{{\gamma}}
\newcommand{\Gam}{{\Gamma}}

\newcommand{\sig}{{\sigma}}
\newcommand{\al}{{\alpha}}
\newcommand{\be}{{\beta}}

\newcommand{\La}{{\Lambda}}

\newcommand{\bbE}{{\mathbb E}}

\newcommand{\bbP}{{\mathbb P}}
\newcommand{\bbR}{{\mathbb R}}

\newcommand{\bbI}{{\mathbb I}}
\newcommand{\bfz}{{\bf z}}
\newcommand{\bfu}{{\bf u}}
\swapnumbers
\theoremstyle{plain}

\numberwithin{equation}{section}

\begin{document}
\title[]{Hedging of game options in discrete markets with transaction costs}%
 \vskip 0.1cm
 \author{ Yuri Kifer\\
\vskip 0.1cm
Institute of Mathematics\\
Hebrew University\\
Jerusalem, Israel}%
\address{
 Institute of Mathematics, The Hebrew University, Jerusalem 91904, Israel}%
\email{  kifer@math.huji.ac.il}%

\thanks{Partially supported by the ISF grant no. 82/10}
\subjclass[2000]{Primary 91B28:  Secondary: 60G40, 91B30 }%
\keywords{ game options, transaction costs, superhedging, shortfall risk, 
Dynkin games }%

 \date{\today}
\begin{abstract}\noindent
 We construct algorithms for computation of prices and superhedging strategies
 for game options in general
 discrete time markets with transaction costs both from seller's (upper
 arbitrage free price) and buyer's (lower arbitrage free price) points of view.
\end{abstract}
%\footnotetext[1]{}
\maketitle
\markboth{ Y.Kifer}{Game options with transaction costs}
\renewcommand{\theequation}{\arabic{section}.\arabic{equation}}
\pagenumbering{arabic}

\section{Introduction}\label{sec1}
Game options introduced in \cite{Ki} were studied by now in scores of
papers but their investigation for markets with transaction costs remains
on its initial stages (see \cite{Do1} and \cite{Do2}). In this paper we
 extend to the game options case the
 theory of pricing and hedging of options in general discrete markets with
 proportional transaction costs in the form which was previously developed for
 American options case in \cite{CJ}, \cite{BT} and \cite{RZ}. 
 
 It is well known that pricing in markets with transaction costs becomes 
 somewhat similar to pricing in incomplete markets so that hedging arguments 
 from the option seller's and buyer's point of view lead to different prices
 which determine the whole range of arbitrage free prices. An interesting but
 not surprising feature of game options is the almost complete symmetry
 between the seller's and the buyer's pricing approaches which will be
 demonstrated clearly in the statements of results and their proofs in this
 paper.
 
 We will derive 
 representations both for upper (ask, seller's) and lower (bid, buyer's) 
 hedging prices of game options with transaction costs and will exhibit
 dynamical programming type algorithms for their computation, as well, as
 for computation of corresponding seller's and buyer's superhedging strategies.
 We demonstrate for game options only analogues of some of representations
 and algorithms from \cite{RZ} concentrating on those which allow symmetric
 expositions of seller's and buyer's cases. Others can be obtained in a 
 similar way and
 their inclusion here would make this paper too overloaded and more difficult
 to read.
 
 Superhedging requires an option seller to invest a large sum into his
 hedging portfolio and in some circumstances he may prefer to accept some
  risk setting up a portfolio with smaller initial amount. In fact, we will
  define in Section \ref{sec5} the shortfall risk both for the seller and 
  the buyer showing that also in the study of the risk for game options 
  the symmetry between these two market participants can be preserved.
  Dynamical programming type algorithms for computation of shortfall risks
  and corresponding partial hedging strategies in the setup of game options
  with transaction costs can be obtained similarly to \cite{Do1} where they
  were derived for American options in binomial markets with the motivation
   to approximate the
  shortfall risk of the continuous time Black-Scholes model.

  \section{Preliminaries}\label{sec2}

We will deal with the same market model as in \cite{RZ} which consists of a 
finite probability space $\Om$ with the $\sig$-field $\cF=2^\Om$ of all 
subspaces of $\Om$ and a probability measure $\bbP$ on $\cF$ giving a 
positive weight $\bbP(\om)$
  to each $\om\in\Om$. The setup includes also a filtration $\{\emptyset,
  \Om\}=\cF_0\subset\cF_1\subset ...\subset\cF_T=\cF$ where $T$ is a positive
   integer called the time horizon. It is convenient to denote by $\Om_t$ the
   set of atoms in $\cF_t$ so that any $\cF_t$-measurable random variable 
   (vector) $Z$ can be identified with a function (vector function) defined on
    $\Om_t$ and its value at $\mu\in\Om_t$ will be denoted either by $Z(\mu)$
    or by $Z^\mu$.
    The points of $\Om_t$ can be viewed as vertices of a tree so that an arrow
    is drawn from $\mu\in\Om_t$ to $\nu\in\Om_{t+1}$ if $\nu\subset\mu$.
     
     The market model consists of a risk-free bond and a risky stock. 
     Without loss of generality, we can assume that all prices are discounted
     so that the bond price equals 1 all the time and a position in bonds is
     identified with cash holding. On the other hand, the shares of the stock
     can be traded which involves proportional transaction costs. This will be
     represented by bid-ask spreads, i.e. shares can be bought at an ask price
     $S^a_t$ or sold at the bid price $S^b_t$, where $S_t^s\geq S^b_t>0,\,
     t=0,1,...,T$ are processes adapted to the filtration $\{\cF_t\}_{t=o}^T$.
     
     The liquidation value at time $t$ of a portfolio $(\gam,\del)$ consisting
     of an amount $\gam$ of cash (or bond) and $\del$ shares of the stock
     equals
     \begin{equation}\label{2.1}
     \te_t(\gam,\del)=\gam+S^b_t\del^+-S^a_t\del^-
     \end{equation}
     which in case $\del<0$ means that a potfolio owner should spend the
     amount $S^a_t\del^-$ in order to close his short position. Observe
     that fractional numbers of shares are allowed here so that both $\gam$
     and $\del$ in a portfolio $(\gam,\del)$ could be, in priciple, any real
     numbers. By definition, a self-financing portfolio strategy is a 
     predictable
     process $(\al_t,\be_t)$ representing positions in cash (or bonds) and
     stock at time $t,\, t=0,1,...,T$ such that 
     \begin{equation}\label{2.2}
     \te_t(\al_t-\al_{t+1},\be_t-\be_{t+1})\geq 0\quad\forall t=0,1,...,T-1
     \end{equation}
     and the set of all such portfolio strategies will be denoted by $\Phi$.
     
     Recall that a game (or Israeli) option (contingent claim) introduced
     in \cite{Ki} is defined as a contract between its seller and buyer such
     that both have the right to exercise it at any time up to a maturity 
     date (horizon) $T$. If the buyer exercises the contract at time $t$, then
     he receives from the seller a payment $Y_t$ while if the latter 
     exercises at time $t$ then his payment to the buyer becomes $X_t\geq
     Y_t$ and the difference $\Del_t=X_t-Y_t$ is interpreted as a penalty for 
     the contract cancellation. In the presence of transaction costs there
     is a difference whether we stipulate that the option to be settled in cash 
     or both in cash and shares of stock while in the former case an assumption
     concerning transaction costs in the process of portfolio liquidation
     should be made. We adopt here the setup where the payments $X_t$ and $Y_t$
     are made both in cash and in shares of the stock and transaction costs
     take place always when a portfolio adjustment occurs. Thus, the payments
     are, in fact, adapted random 2-vectors $X_t=(X_t^{(1)},X_t^{(2)})$ and
     $Y_t=(Y_t^{(1)},Y_t^{(2)})$ where the first and the second coordinates
     represent, respectively, a cash amount to be payed and a number of stock
     shares to be delivered and as we allow also fractional numbers of shares
     both coordinates can take on any nonnegative real value. The inequality
     $X_t\geq Y_t$ in the zero transaction costs case is replaced in our 
     present setup by
     \begin{equation}\label{2.3}
     \Del_t=\te_t(X^{(1)}_t-Y_t^{(1)},X_t^{(2)}-Y_t^{(2)})\geq 0
     \end{equation}
     and $\Del_t$ is interpreted as a cancellation penalty. We impose also 
     a natural assumption that $X^{(1)}_T=Y^{(1)}_T$ and $X^{(2)}_T=Y^{(2)}_T$,
     i.e. on the maturity date there is no penalty. Therefore, if
     the seller cancells the contract at time $s$ while the buyer exercises
     at time $t$ the former delivers to the latter a package of cash and
     stock shares which can be represented as a 2-vector in the form
     \begin{equation}\label{2.4}
     Q_{s,t}=(Q^{(1)}_{s,t},Q^{(2)}_{s,t})=X_s\bbI_{s<t}+Y_t\bbI_{t\leq s}
     \end{equation}
     where $\bbI_A=1$ if an event $A$ occurs and $\bbI_A=0$ if not. It will
     be convenient to allow the payment components $X_t^{(1)},\, X_t^{(2)}$
     and $Y_t^{(1)},\, Y_t^{(2)}$ to take on any real (and not only
     nonnegative) values which will enable us to demonstrate complete 
     duality (symmetry) between the seller's and the buyer's positions.
     
     A pair $(\sig,\pi)$ of a stopping time $\sig\leq T$ and of a 
     self-financing strategy $\pi=(\al_t,\be_t)^T_{t=0}$ will be called a
     superhedging strategy for the seller of the game option with a payoff
     given by (\ref{2.4}) if for all $t\leq T$,
     \begin{equation}\label{2.5}
     \te_{\sig\wedge t}(\al_{\sig\wedge t}-Q^{(1)}_{\sig,t},\,
     \be_{\sig\wedge t}-Q^{(2)}_{\sig,t})\geq 0
     \end{equation}
     where, as usual, $c\wedge d=\min(c,d)$ and $c\vee d=\max(c,d)$. The
     seller's (ask or upper hedging) price $V^a$ of a game option is defined 
     as the infimum of initial amounts required to start a superhedging
     strategy for the seller. Since in order to get $\al_0$ amount of
     cash and $\be_0$ shares of stock at time 0 the seller should spend
     \begin{equation}\label{2.6}
     -\te_0(-\al_0,-\be_0)=\al_0+\be_0^+S_0^a-\be_0^-S^b_0
     \end{equation}
     in cash, we can write
     \begin{equation}\label{2.7}
     V^a=\inf_{\sig,\pi}\{-\te_0(-\al_0,-\be_0):\, (\sig,\pi)\,\,\mbox{with}
     \,\,\pi=(\al_t,\be_t)_{t=0}^T\,\,\mbox{being a superhedging strategy for 
     the seller}\}.
     \end{equation}
     
     On the other hand, the buyer may borrow from a bank an amount
     $\te_0(-\al_0,-\be_0)$ to purchase a game option with the payoff
     (\ref{2.4}) and starting with the negative valued portfolio $(\al_0,
     \be_0)$ to manage a self-financing strategy $\pi=(\al_t,\be_t)^T_{t=0}$
     so that for a given stopping time $\tau\leq T$ and all $s\leq T$,
     \begin{equation}\label{2.8}
      \te_{s\wedge\tau}(\al_{s\wedge\tau}+Q^{(1)}_{s,\tau},\,\be_{s\wedge\tau}
      +Q^{(2)}_{s,\tau})\geq 0.
     \end{equation}
     In this case the pair $(\tau,\pi)$ will be called a superhedging strategy
     for the buyer. The buyer's (bid or lower hedging) price $V^b$ of the
     game option above is defined as the supremum of initial bank loan 
     required to purchase this game option and to manage a superhedging
     strategy for the buyer. Thus,
     \begin{equation}\label{2.9}
     V^b=\sup_{\tau,\pi}\{\te_0(-\al_0,-\be_0):\, (\tau,\pi)\,\,\mbox{with}
     \,\,\pi=(\al_t,\be_t)_{t=0}^T\,\,\mbox{being a superhedging strategy for 
     the buyer}\}.
     \end{equation}
     It follows from the representations of Theorem \ref{thm3.1} below that
     $V^a\geq V^b$.
     
     The goal of this paper is to obtain representations of $V^a$ and
     $V^b$ in the form of $\inf\sup$ expressions and to
     construct backward and forward induction algorithms for computation
     both of these prices and of corresponding superhedging strategies.
     These results will be stated precisely in the next section.
     As in the case of American options with transaction costs in \cite{CJ},
     \cite{BT} and \cite{RZ} precise statements of our results involve the
     notion of randomized stopping times and approximate martingales which
     will be introduced in the next section.

     \section{Superhedging and price representations: statements}\label{sec3}

First, we recall the notion of a randomized stopping time (see \cite{CJ},
\cite{BT}, \cite{RZ} and references there) which is defined as a nonnegative
adapted process $\chi$ such that $\sum_{t=0}^T\chi_t=1$. The set of all 
randomized stopping times will be denoted by $\cX$ while the set of all 
usual or pure stopping times will be denoted by $\cT$. It will be convenient
to identify each pure stopping time $\tau$ with a randomized stopping time
$\chi^\tau$ such that $\chi^\tau_t=\bbI_{\{\tau=t\}}$ for any $t=0,1,...,T$,
so that we could write $\cT\subset\cX$. For any adapted process $Z$ and each
randomized stopping time $\chi$ the time-$\chi$ value of $Z$ is defined by
\begin{equation}\label{3.1}
Z_\chi=\sum_{t=0}^T\chi_tZ_t.
\end{equation}

Considering a game option with a payoff given by (\ref{2.4}) we
write also
\begin{equation}\label{3.2}
Q_{\chi,\tilde\chi}=\sum_{s,t=0}^T\chi_s\tilde\chi_tQ_{s,t}
\end{equation}
which is the seller's payment to the buyer when the former cancells and
the latter exercises at randomized stopping times $\chi$ and $\tilde\chi$,
respectively. In particular, if $\sig$ and $\tau$ are pure stopping times
then
\begin{equation}\label{3.3}
Q_{\chi,\chi^\tau}=\sum_{s=0}^T\chi_sQ_{s,\tau}\,\,\mbox{and}\,\, Q_{\chi^\sig,
\chi}=\sum_{t=0}^T\chi_tQ_{\sig,t}.
\end{equation}
We can also define the "minimum" and the "maximum" of two randomized stopping
times $\chi$ and $\tilde\chi$ which are randomized stopping times $\chi\wedge
\tilde\chi$ and $\chi\vee\tilde\chi$ given by
\begin{eqnarray}\label{3.4}
&(\chi\wedge\tilde\chi)_t=\chi_t\sum^T_{s=t}\tilde\chi_s+\tilde\chi_t\sum^T_{s=t+1}
\chi_s\,\,\mbox{and}\\
&(\chi\vee\tilde\chi)_t=\chi_t\sum^t_{s=0}\tilde\chi_s+\tilde\chi_t
\sum^{t-1}_{s=0}\chi_s.\nonumber
\end{eqnarray}
In particular, if $\sig$ and $\tau$ are pure stopping times then 
\begin{equation}\label{3.5}
\chi^\sig\wedge\chi^\tau=\chi^{\sig\wedge\tau},\,\,\chi^\sig\vee\chi^\tau=
\chi^{\sig\vee\tau}
\end{equation}
and for any adapted process $Z$,
\begin{equation}\label{3.6}
Z_{\chi\wedge\chi^\tau}=\sum_{s=0}^T\chi_sZ_{s\wedge\tau},\,\,
Z_{\chi^\sig\wedge\chi}=\sum_{t=0}^T\chi_tZ_{\sig\wedge t}
\end{equation}
and similarly for $\chi\vee\chi^\tau$ and $\chi^\sig\vee\chi$.

Next, we introduce the notion of an approximate martingale which is defined
for any randomized stopping time $\chi$ as a pair $(P,S)$ of a probability
 measure $P$ on $\Om$ and of an adapted process $S$ such that for each
 $t=0,1,...,T$,
 \begin{equation}\label{3.7}
 S_t^b\leq S_t\leq S_t^a\,\,\mbox{and}\,\,\chi_{t+1}^*S^b_t\leq\bbE_P(
 S^{\chi^*}_{t+1}|\cF_t)\leq\chi_{t+1}^*S^a_t
 \end{equation}
 where $\bbE_P$ is the expectation with respect to $P$,
 \begin{equation}\label{3.8}
 \chi^*_t=\sum_{s=t}^T\chi_s,\, Z_t^{\chi^*}=\sum_{s=t}^T\chi_sZ_s,\,
 \chi^*_{T+1}=0\,\,\mbox{and}\,\, Z_{T+1}^{\chi^*}=0.
 \end{equation}
 Given a randomized stopping time $\chi$ the space of corresponding
  approximate martingales $(P,S)$ will be denoted by $\bar\cP(\chi)$ and
  we denote by $\cP(\chi)$ the subspace of $\bar\cP(\chi)$ consisting of
  pairs $(P,S)$ with $P$ being equivalent to the original (market)
  probability $\bbP$.
  
  Now we can formulate some of our results which exhibit ask and bid price
  representations for game options.
  \begin{theorem}\label{thm3.1} In the above notations,
  \begin{eqnarray}\label{3.9}
  &V^a=\min_{\sig\in\cT}\max_{\chi\in\cX}\max_{(P,S)\in\bar\cP(\chi)}
  \bbE_P\big(Q^{(1)}_{\sig,\cdot}+SQ^{(2)}_{\sig,\cdot}\big)_\chi\\
  &=\min_{\sig\in\cT}\max_{\chi\in\cX}\sup_{(P,S)\in\cP(\chi)}
  \bbE_P\big(Q^{(1)}_{\sig,\cdot}+SQ^{(2)}_{\sig,\cdot}\big)_\chi\nonumber
  \end{eqnarray}
and
\begin{eqnarray}\label{3.10}
  &V^b=\max_{\tau\in\cT}\min_{\chi\in\cX}\min_{(P,S)\in\bar\cP(\chi)}
  \bbE_P\big(Q^{(1)}_{\cdot,\tau}+SQ^{(2)}_{\cdot,\tau}\big)_\chi\\
  &=\max_{\tau\in\cT}\min_{\chi\in\cX}\inf_{(P,S)\in\cP(\chi)}
 \bbE_P\big(Q^{(1)}_{\cdot,\tau}+SQ^{(2)}_{\cdot,\tau}\big)_\chi\nonumber
  \end{eqnarray}
  where $Q^{(1)}_{\sig,\cdot},\, Q^{(2)}_{\sig,\cdot}$ and
   $Q^{(1)}_{\cdot,\tau},\, Q^{(2)}_{\cdot,\tau}$ denote functions on
   $\{ 0,1,...,T\}$ whose values at $t$ are obtained by replacing $\cdot$
   by $t$. 
\end{theorem}

In order to exhibit dynamical programming algorithms for computation of $V^a$
and $V^b$ and induction algorithms producing corresponding superhedging
strategies we have to introduce first some convex analysis notions and
notations (see \cite{Ro} and \cite{RZ} for more details). Denote by
$\Te$ the family of functions $f:\,\bbR\to\bbR\cup\{-\infty\}$ such that
either $f\equiv -\infty$ or $f$ is a (finite) real valued polyhedral
(continuous piecewise linear with finite number of segments) function. If
$f,g\in\Te$ then, clearly, $f\wedge g,\, f\vee g\in\Te$. The epigraph of
$f\in\Te$ is defined by epi$(f)=\{(x,y)\in\bbR^2:\, x\geq f(y)\}$. For any
$c\geq d$ the function $h_{[d,c]}(y)=cy^--dy^+$, clearly, belongs to $\Te$.
Observe that the self-financing condition (\ref{2.2}) can be rewritten in
the form
\begin{equation}\label{3.11}
(\al_t-\al_{t+1},\be_t-\be_{t+1})\in\,\mbox{epi}(h_{[S_t^b,S_t^a]}).
\end{equation}
For each $f\in\Te$ and $c\geq d$ there exists a unique function 
gr$_{[d,c]}(f)\in\Te$ such that
\begin{equation}\label{3.12}
\mbox{epi(gr}_{[d,c]}(f))=\mbox{epi}(h_{[d,c]})+\mbox{epi}(f).
\end{equation}

It is clear from (\ref{3.11}) and (\ref{3.12}) that portfolios in 
epi(gr$_{[S^b_t,S^a_t]}(f))$ are precisely those which can be rebalanced in 
a self-financing manner at time $t$ to yield a portfolio in epi$(f)$. Denote
by $\La$ the family of all convex functions in $\Te$  and by $\Gam$ the 
family of concave functions $v:\,\bbR\to\bbR\cup\{-\infty\}$ which are
polyhedral on their essential domain dom$(v)=\{ x\in\bbR:\, v(x)>-\infty\}$.
For any $f\in\La$ the convex duality sais that
\begin{equation}\label{3.13}
f^*(x)=\inf_{y\in\bbR}(f(y)+xy)\in\Gam\,\,\mbox{and}\,\, f(y)=\sup_{x\in\bbR}
(f^*(x)-xy),
\end{equation}
the infimum and the supremum above are attained whenever they are finite.

For any $y\in\bbR,\,\mu\in\Om_t$ and $t=0,1,...,T$ define $q_t^a(y)=
q_t^a(\mu,y)$, $q_t^b(y)=q_t^b(\mu,y)$, $r_t^a(y)=r_t^a(\mu,y)$, $r_t^b(y)
=r_t^b(\mu,y)$ by
\begin{eqnarray*}
&q^a_t(y)=X^{(1)}_t+h_{[S^b_t,S^a_t]}(y-X_t^{(2)}),\,\, 
r^a_t(y)=Y^{(1)}_t+h_{[S^b_t,S^a_t]}(y-Y_t^{(2)})\\
&q^b_t(y)=-X^{(1)}_t+h_{[S^b_t,S^a_t]}(y+X_t^{(2)}),\,\, 
r^b_t(y)=-Y^{(1)}_t+h_{[S^b_t,S^a_t]}(y+Y_t^{(2)})
\end{eqnarray*}
with $h_{[d,c]}$ the same as in (\ref{3.11}) and (\ref{3.12}). Observe that 
if $c\geq d\geq 0$ then either $h_{[d,c]}\equiv 0$ or $h_{[d,c]}$ is a monotone
 decreasing function, and so
\begin{equation}\label{3.14}
q^a_t\geq r^a_t\,\,\mbox{and}\,\, q^b_t\leq r^b_t.
\end{equation}
Introduce also
\begin{eqnarray}\label{3.15}
&G_{s,t}(y)=Q^{(1)}_{s,t}+h_{[S^b_{s\wedge t},S^a_{s\wedge t}]}(y-
Q_{s,t}^{(2)})=q^a_s(y)\bbI_{s<t}+r^a_t(y)\bbI_{t\leq s}\,\,\mbox{and}\\
&H_{s,t}(y)=-Q^{(1)}_{s,t}+h_{[S^b_{s\wedge t},S^a_{s\wedge t}]}(y+
Q_{s,t}^{(2)})=q^b_s(y)\bbI_{s<t}+r^b_t(y)\bbI_{t\leq s}.\nonumber
\end{eqnarray}
Clearly, the superhedging conditions (\ref{2.5}) of the seller and (\ref{2.8})
of the buyer are equivalent to
\begin{equation}\label{3.16}
(\al_{\sig\wedge t},\be_{\sig\wedge t})\in\,\mbox{epi}(G_{\sig,t})\,\,
\mbox{for all}\,\, t=0,1,...,T\,\,\,\mbox{and}
\end{equation}
\begin{equation}\label{3.17}
(\al_{s\wedge\tau},\be_{s\wedge\tau})\in\,\mbox{epi}(H_{s,\tau})\,\,
\mbox{for all}\,\, s=0,1,...,T,
\end{equation}
respectively. Observe also that
\begin{equation}\label{3.18}
 q_t^a(0)=-q^b_t(0)=\te_t(X_t^{(1)},
X_t^{(2)})\,\,\mbox{and}\,\, r_t^a(0)=-r^b_t(0)=\te_t(Y_t^{(1)},Y_t^{(2)}).
\end{equation}
We recall that $X^{(1)}_T=Y^{(1)}_T$ and $X^{(2)}_T=Y^{(2)}_T$, and so 
$q^a_T=r^a_T$ and $q^b_T=r^b_T$. The following result provides dynamical
programming algorithms for ask and bid prices computations.

\begin{theorem}\label{thm3.2} (i) For any $x\in\bbR$, $\mu\in\Om_T$ and 
$\sig\in\cT$ define
\begin{equation}\label{3.19}
z^\mu_T(x)=w_T^\mu(x)=r^a_T(\mu,x)\,\,\mbox{and}\,\, z^\mu_{\sig,T}(x)=w_{\sig,T}^\mu
(x)=G_{\sig,T}(\mu,x).
\end{equation}
Next, for $t=1,2,...,T$ and each $\mu\in\Om_{t-1}$ define by backward
induction
\begin{eqnarray}\label{3.20}
&\bfz^\mu_{t-1}=\max_{\nu\subset\mu,\,\nu\in\Om_t}z_t^\nu,\,\,
\bfz^\mu_{\sig,t-1}=\max_{\nu\subset\mu,\,\nu\in\Om_t}z_{\sig,t}^\nu,\\
&w^\mu_{t-1}=\mbox{gr}_{[S^b_{t-1}(\mu),S^a_{t-1}(\mu)]}(\bfz^\mu_{t-1}),\,\,
w^\mu_{\sig,t-1}=\mbox{gr}_{[S^b_{t-1}(\mu),S^a_{t-1}(\mu)]}(\bfz^\mu_{\sig,t-1}
),\nonumber\\
&z^\mu_{t-1}(x)=\min\big(q^a_{t-1}(\mu,x),\max(r^a_{t-1}(\mu,x),w^\mu_{t-1}(x))
\big)\,\,\mbox{and}\,\,
z^\mu_{\sig,t-1}(x)=\max(G_{\sig,t-1}(\mu,x),w^\mu_{\sig,t-1}(x)).
\nonumber\end{eqnarray}
Then $z_0(0)=\min_{\sig\in\cT}z_{\sig,0}(0)=V^a$.

(ii) For any $x\in\bbR$, $\mu\in\Om_T$ and $\tau\in\cT$ define
\begin{equation}\label{3.21}
u^\mu_T(x)=v^\mu_T(x)=r^b_T(\mu,x)\,\,\mbox{and}\,\, u^\mu_{T,\tau}(x)=
v_{T,\tau}^\mu(x)=H_{T,\tau}(\mu,x).
\end{equation}
Next, for $t=1,2,...,T$ and each $\mu\in\Om_{t-1}$ define by the backward
 induction
 \begin{eqnarray}\label{3.22}
&\bfu^\mu_{t-1}=\max_{\nu\subset\mu,\,\nu\in\Om_t}u_t^\nu,\,\,
\bfu^\mu_{t-1,\tau}=\max_{\nu\subset\mu,\,\nu\in\Om_t}u_{t,\tau}^\nu,\\
&v^\mu_{t-1}=\mbox{gr}_{[S^b_{t-1}(\mu),S^a_{t-1}(\mu)]}(\bfu^\mu_{t-1}),\,\,
v^\mu_{t-1,\tau}=\mbox{gr}_{[S^b_{t-1}(\mu),S^a_{t-1}(\mu)]}(\bfu^\mu_{t-1,
\tau}),\nonumber\\
&u^\mu_{t-1}(x)=\min\big(r^b_{t-1}(\mu,x),\max(q^b_{t-1}(\mu,x),v^\mu_{t-1}(x))
\big)\,\,\mbox{and}\,\, u^\mu_{t-1,\tau}(x)=\max(H_{t-1,\tau}(\mu,x),
v^\mu_{t-1,\tau}(x)).
\nonumber\end{eqnarray}
Then $u_0(0)=\max_{\tau\in\cT}u_{0,\tau}(0)=-V^b$.
\end{theorem}

Next, we describe inductive constructions of superhedging seller's and 
buyer's strategies using the functions $z_t$ and $u_t$, $t=0,...,T$
constructed in Theorem \ref{thm3.2}.
\begin{theorem}\label{thm3.3} (i) Construct by induction a sequence of (pure) 
stopping times $\sig_t\in\cT$ and a self-financing strategy $(\al,\be)$ such
that
\begin{equation}\label{3.23}
(\al_t,\be_t)\in\mbox{epi}(z_t)\setminus\mbox{epi}(q^a_t)\,\,\mbox{on}\,\,
\{t<\sig_t\}
\end{equation}
for each $t=0,1,...,T$ in the following way. First, take any 
$\cF_0$-measurable portfolio $(\al_0,\be_0)\in\mbox{epi}(z_0)$ and set
\begin{equation*}
   \sig_0=\left\{\begin{array}{ll}
  0 &\mbox{if}\,\, (\al_0,\be_0)\in\mbox{epi}(q_0^a)\\
  T &\mbox{if}\,\, (\al_0,\be_0)\notin\mbox{epi}(q_0^a).
  \end{array}\right.
  \end{equation*}
  Suppose that an $\cF_t$-measurable portfolio $(\al_t,\be_t)\in
  \mbox{epi}(z_t)$ and a stopping time $\sig_t\in\cT$ have already been
   constructed for some $t=0,1,...T-1$ so that (\ref{3.23}) holds true. 
   By (\ref{3.12}) and (\ref{3.20}),
   \[
   (\al_t,\be_t)\in\mbox{epi}(w_t)=\mbox{epi}(h_{[S^b_t,S^a_t]})+\mbox{epi}
   (\bfz_t)\,\,\mbox{on}\,\,\{t<\sig_t\},
   \]
   and so there exists an $\cF_t$-measurable portfolio $(\al_{t+1},\be_{t+1})$
   such that
   \[
   (\al_{t+1},\be_{t+1})\in\mbox{epi}(\bfz_t),
   \,\, (\al_t-\al_{t+1},\be_t-\be_{t+1})\in\mbox{epi}(h_{[S^b_t,S^a_t]})\,\,
   \mbox{on}\,\,\{t<\sig_t\}
   \]
   and $(\al_{t+1},\be_{t+1})=(\al_t,\be_t)$ on $\{ t\geq\sig_t\}$ which
   provides the self-financing condition (\ref{3.11}) both on $\{ t<\sig_t\}$
   and on $\{ t\geq\sig_t\}$. By (\ref{3.20}) it follows also that $(\al_{t+1},
   \be_{t+1})\in\,\mbox{epi}(z_{t+1})$ on $\{ t<\sig_t\}\supset
   \{ t+1<\sig_{t+1}\}$. Set
   \begin{equation*}
   \sig_{t+1}=\left\{\begin{array}{ll}
  \sig_t &\mbox{if}\,\, t\geq \sig_t\\
  t+1 &\mbox{if}\,\, t<\sig_t\,\,\mbox{and}\,\,(\al_{t+1},\be_{t+1})\in
  \mbox{epi}(q_{t+1}^a)\\
  T &\mbox{if}\,\, t<\sig_t\,\,\mbox{and}\,\,(\al_{t+1},\be_{t+1})\notin
  \mbox{epi}(q_{t+1}^a).
  \end{array}\right.
  \end{equation*}
  Finally, set $\sig=\sig_T\in\cT$. Then the pair $(\sig,\pi)$ with
  $\pi=(\al,\be)$ constructed by the above algorithm with $(\al_0,\be_0)=
  (V^a,0)$ is a superhedging strategy for the seller.  
  
  (ii) Construct by induction a sequence of (pure) 
stopping times $\tau_t\in\cT$ and a self-financing strategy $(\al,\be)$ such
that
\begin{equation}\label{3.24}
(\al_t,\be_t)\in\mbox{epi}(u_t)\setminus\mbox{epi}(r^b_t)\,\,\mbox{on}\,\,
\{t<\tau_t\}
\end{equation}
for each $t=0,1,...,T$ in the following way. First, take any 
$\cF_0$-measurable portfolio $(\al_0,\be_0)\in\mbox{epi}(u_0)$ and set
\begin{equation*}
   \tau_0=\left\{\begin{array}{ll}
  0 &\mbox{if}\,\, (\al_0,\be_0)\in\mbox{epi}(r_0^b)\\
  T &\mbox{if}\,\, (\al_0,\be_0)\notin\mbox{epi}(r_0^b).
  \end{array}\right.
  \end{equation*}
  Suppose that an $\cF_t$-measurable portfolio $(\al_t,\be_t)\in
  \mbox{epi}(u_t)$ and a stopping time $\tau_t\in\cT$ have already been
   constructed for some $t=0,1,...T-1$ so that (\ref{3.23}) holds true. 
   By (\ref{3.12}) and (\ref{3.22}),
   \[
   (\al_t,\be_t)\in\mbox{epi}(v_t)=\mbox{epi}(h_{[S^b_t,S^a_t]})+\mbox{epi}
   (\bfu_t)\,\,\mbox{on}\,\,\{t<\tau_t\},
   \]
   and so there exists an $\cF_t$-measurable portfolio $(\al_{t+1},\be_{t+1})$
   such that
   \[
   (\al_{t+1},\be_{t+1})\in\mbox{epi}(\bfu_t),
   \,\, (\al_t-\al_{t+1},\be_t-\be_{t+1})\in\mbox{epi}(h_{[S^b_t,S^a_t]})\,\,
   \mbox{on}\,\,\{t<\tau_t\}
   \]
   and $(\al_{t+1},\be_{t+1})=(\al_t,\be_t)$ on $\{ t\geq\tau_t\}$ which
   provides the self-financing condition (\ref{3.11}) both on $\{ t<\tau_t\}$
   and on $\{ t\geq\tau_t\}$. By (\ref{3.22}) it follows also that $(\al_{t+1},
   \be_{t+1})\in\,\mbox{epi}(u_{t+1})$ on $\{ t<\tau_t\}\supset
   \{ t+1<\tau_{t+1}\}$. Set
   \begin{equation*}
   \tau_{t+1}=\left\{\begin{array}{ll}
  \tau_t &\mbox{if}\,\, t\geq \tau_t\\
  t+1 &\mbox{if}\,\, t<\tau_t\,\,\mbox{and}\,\,(\al_{t+1},\be_{t+1})\in
  \mbox{epi}(r_{t+1}^b)\\
  T &\mbox{if}\,\, t<\tau_t\,\,\mbox{and}\,\,(\al_{t+1},\be_{t+1})\notin
  \mbox{epi}(r_{t+1}^b).
  \end{array}\right.
  \end{equation*}
  Finally, set $\tau=\tau_T\in\cT$. Then the pair $(\tau,\pi)$ with
  $\pi=(\al,\be)$ constructed by the above algorithm with $(\al_0,\be_0)=
  (-V^b,0)$ is a superhedging strategy for the buyer.
   \end{theorem}
   
   \section{Superhedging and price representations: proofs}\label{sec4}

We start with the following result. 
\begin{lemma}\label{lem4.1} (i) The following assertions are equivalent:

(a) $(\gam,\del)\in$ epi$(z_0)$;

(b) There exists a self-financing strategy $(\al,\be)\in\Phi$ and a 
stopping time $\sig\in\cT$ such that $(\al_0,\be_0)=(\gam,\del)$ and
(\ref{3.16}) holds true;

(c) There exists a superhedging strategy $(\sig,\pi),\,\sig\in\cT,\,\pi=
(\al_t,\be_t)_{t=0}^T\in\Phi$ for the seller such that $(\al_0,\be_0)=
(\gam,\del)$.

(ii) The following assertions are equivalent:

(d) $(\gam,\del)\in$ epi$(u_0)$;

(e) There exists a self-financing strategy $(\al,\be)\in\Phi$ and a 
stopping time $\tau\in\cT$ such that $(\al_0,\be_0)=(\gam,\del)$ and
(\ref{3.17}) holds true;

(f) There exists a superhedging strategy $(\tau,\pi),\,\tau\in\cT,\,\pi=
(\al_t,\be_t)_{t=0}^T\in\Phi$ for the buyer such that $(\al_0,\be_0)=
(\gam,\del)$.
\end{lemma}
\begin{proof} 
(i) Suppose that $(\gam,\del)\in$ epi$(z_0)$. Set $(\al_0,\be_0)=(\gam,\del)$ 
and construct a self-financing strategy $\pi=(\al_t,\be_t)_{t=0}^T$ and
a stopping time $\sig$ inductively as described in Theorem \ref{thm3.3}(i).
In order to obtain (\ref{3.16}) we show first that
\begin{equation}\label{4.1}
(\al_\sig,\be_\sig)\in\,\mbox{epi}(q^a_\sig).
\end{equation}
Indeed, if $\sig=\sig_T=t<T$ then by the construction $\sig_T=\sig_{T-1}=
\cdots =\sig_t,\, t-1<\sig_{t-1}$ and $(\al_t,\be_t)\in$ epi$(q^a_t)$.
If $\sig=\sig_T=T$ then we also obtain that
 \begin{equation}\label{4.2}
(\al_T,\be_T)\in\,\mbox{epi}(q^a_T)
\end{equation}
 since for otherwise by the construction we must have $\sig_{T-1}=\sig_T=T$.
 Then $T-1<\sig_{T-1}$ and on this set (event) by construction
$(\al_T,\be_T)\in$ epi$(\bfz_{T-1})\subset$ epi$(z_T)=$epi$(q^a_T)$ since 
$z_T=r^a_T=q^a_T$, and so (\ref{4.2}) holds true completing the proof of
(\ref{4.1}), as well. Since $G_{\sig,t}=q^a_\sig$ if $\sig<t$ by (\ref{3.15})
then (\ref{3.16}) follows from (\ref{4.1}) when $\sig<t$. If $\sig=t$ then
$G_{\sig,t}=r^a_t=r^a_\sig\leq q^a_\sig$, and so (\ref{4.1}) implies 
(\ref{3.16}) in this case, as well. Next, assume that $t<\sig=\sig_T$ and 
observe that by the construction $\{ t<\sig\}\subset\{ t<\sig_t\}$ for all $t=
0,1,...,T$. On $\{ t<\sig_t\}$ we have $(\al_t,\be_t)\in$ epi$(z_t)$. But by 
(\ref{3.20}) either $z_t=q_t^a\geq r^a_t$ or $z_t=\max(r^a_t,w_t)\geq r^a_t$
and in both cases $(\al_t,\be_t)\in$ epi$(r^a_t)$ which yields (\ref{3.16})
since $G_{\sig,t}=r^a_t$ when $t<\sig$.

Next, suppose that (b) holds true. Then (c) follows since (\ref{3.16}) is
equivalent to the seller's superhedging condition (\ref{2.5}).

Suppose now that $(\sig,\pi),\,\sig\in\cT,\,\pi=(\al_t,\be_t)_{t=0}^T$ is a
superhedging strategy for the seller such that $(\al_0,\be_0)=(\gam,\del)$.
Then (\ref{3.16}) holds true. We prove next by the backward induction that
\begin{equation}\label{4.3}
(\al_{\sig\wedge t},\be_{\sig\wedge t})\in\,\mbox{epi}(z_{\sig\wedge t}).
\end{equation}
Observe that either $G_{\sig,T}=r^a_T=z_T$ if $\sig=T$ or $G_{\sig,T}=q^a_\sig
\geq z_\sig$ if $\sig<T$. In both cases epi$(G_{\sig,T})\subset$ 
epi$(z_{\sig\wedge T})$ and (\ref{4.3}) with $t=T$ follows from (\ref{3.16}).
Now suppose that (\ref{4.3}) holds true for some
$t\in\{ 1,2,...,T\}$, i.e. $(\al_t,\be_t)\in$ epi$(z_t)$. Since $(\al_{\sig
\wedge t},\be_{\sig\wedge t})$ is $\cF_{\sig\wedge (t-1)}$-measurable it follows
by (\ref{3.20}) that $(\al_{\sig\wedge t},\be_{\sig\wedge t})\in$ epi
$(\bfz_{\sig\wedge (t-1)})$. Because the strategy is self-financing, we have by
(\ref{3.11}) that
\[
(\al_{\sig\wedge (t-1)}-\al_{\sig\wedge t},\,\be_{\sig\wedge (t-1)}-
\be_{\sig\wedge t})\in\,\mbox{epi}(h_{[S^b_{\sig\wedge (t-1)}, 
S^a_{\sig\wedge (t-1)}]}),
\]
obtaining by (\ref{3.20}) that
\[
(\al_{\sig\wedge (t-1)},\,\be_{\sig\wedge (t-1)})\in\,\mbox{epi}(h_{[S^b_{\sig
\wedge (t-1)}, S^a_{\sig\wedge (t-1)}]})+\,\mbox{epi}(\bfz_{\sig\wedge (t-1)})= 
\,\mbox{epi}(w_{\sig\wedge (t-1)}).
\]
Furthermore, by the superhedging condition (\ref{3.16}),
\[
(\al_{\sig\wedge (t-1)},\,\be_{\sig\wedge (t-1)})\in\,\mbox{epi}(G_{\sig,t-1}),
\]
and so,
\begin{eqnarray*}
&(\al_{\sig\wedge (t-1)},\,\be_{\sig\wedge (t-1)})\in\,\mbox{epi}(\max(
G_{\sig,t-1},\, w_{\sig\wedge (t-1)}))\\
&\subset\,\mbox{epi}(\max(
r^a_{\sig\wedge (t-1)},w_{\sig\wedge (t-1)}))\subset\,\mbox{epi}(
z_{\sig\wedge (t-1)})
\end{eqnarray*}
since $z_{\sig\wedge (t-1)}\leq\max(r^a_{\sig\wedge (t-1)},\, 
w_{\sig\wedge (t-1)})$ by (\ref{3.20}), completing the induction step. Now
taking $t=0$ in (\ref{4.3}) we obtain the assertion (a) proving the 
statement (i).

(ii) Suppose that $(\gam,\del)\in$ epi$(u_0)$. Set $(\al_0,\be_0)=(\gam,\del)$ 
and construct a self-financing strategy $\pi=(\al_t,\be_t)_{t=0}^T$ and
a stopping time $\tau$ inductively as described in Theorem \ref{thm3.3}(ii).
In order to obtain (\ref{3.17}) we show first that
\begin{equation}\label{4.4}
(\al_\tau,\be_\tau)\in\,\mbox{epi}(r^b_\tau).
\end{equation}
Indeed, if $\tau=\tau_T=t<T$ then by the construction $\tau_T=\tau_{T-1}=
\cdots =\tau_t,\, t-1<\tau_{t-1}$ and $(\al_t,\be_t)\in$ epi$(r^b_t)$.
If $\tau=\tau_T=T$ then we also obtain that
 \begin{equation}\label{4.5}
(\al_T,\be_T)\in\,\mbox{epi}(r^b_T)
\end{equation}
 since for otherwise by the construction we must have $\tau_{T-1}=\tau_T=T$.
 Then $T-1<\sig_{T-1}$ and on this set (event) by the construction
$(\al_T,\be_T)\in$ epi$(\bfu_{T-1})\subset$ epi$(u_T)=$epi$(r^b_T)$ since 
$u_T=r^b_T$, and so (\ref{4.5}) holds true completing the proof of
(\ref{4.4}), as well. Since $H_{s,\tau}=r^b_\tau$ if $\tau\leq s$ by 
(\ref{3.15})
then (\ref{3.17}) follows from (\ref{4.4}) when $\tau\leq s$. 
 Next, assume that $s<\tau=\tau_T$ and 
observe that by the construction $\{ s<\tau\}\subset\{ s<\tau_s\}$ for all $s=
0,1,...,T$. On $\{ s<\tau_s\}$ we have $(\al_s,\be_s)\in$ epi$(u_s)$. But 
$u_s\geq q^b_s$ by (\ref{3.22}), and so $(\al_s,\be_s)\in$ epi$(q^b_s)$ on
$\{ s<\tau_s\}$, which yields (\ref{3.17}) since $H_{s,\tau}=q^b_s$ when 
$s<\tau$.

Next, suppose that (e) holds true. Then (f) follows since (\ref{3.17}) is
equivalent to the buyer's superhedging condition (\ref{2.8}).

Suppose now that $(\tau,\pi),\,\tau\in\cT,\,\pi=(\al_t,\be_t)_{t=0}^T$ is a
superhedging strategy for the buyer such that $(\al_0,\be_0)=(\gam,\del)$.
Then (\ref{3.17}) holds true. We prove next by the backward induction that
\begin{equation}\label{4.6}
(\al_{s\wedge\tau},\be_{s\wedge\tau})\in\,\mbox{epi}(u_{s\wedge\tau}).
\end{equation}
Observe that either $H_{T,\tau}=r^b_T=u_T$ if $\tau=T$ or $H_{T,\tau}=
r^b_\tau\geq u_\tau$ if $\tau<T$. In both cases epi$(H_{T,\tau})\subset$
epi$(u_{T\wedge\tau})$ and (\ref{4.6}) with $s=T$ follows from (\ref{3.17}).
 Now suppose that (\ref{4.6}) holds true for some
$s\in\{ 1,2,...,T\}$, i.e. $(\al_s,\be_s)\in$ epi$(u_s)$. Since $(\al_{s
\wedge\tau},\be_{s\wedge\tau})$ is $\cF_{(s-1)\wedge\tau}$-measurable it 
follows by (\ref{3.22}) that $(\al_{s\wedge\tau},\,\be_{s\wedge\tau})\in$ 
epi$(\bfu_{(s-1)\wedge\tau})$. Because the strategy is self-financing, we have by
(\ref{3.11}) that
\[
(\al_{(s-1)\wedge\tau}-\al_{s\wedge\tau},\,\be_{(s-1)\wedge\tau}-
\be_{s\wedge\tau})\in\,\mbox{epi}(h_{[S^b_{(s-1)\wedge\tau}, 
S^a_{(s-1)\wedge\tau}]}),
\]
obtaining by (\ref{3.22}) that
\[
(\al_{(s-1)\wedge\tau},\,\be_{(s-1)\wedge\tau})\in\,\mbox{epi}(h_{[S^b_{(s-1)
\wedge\tau}, S^a_{(s-1)\wedge\tau}]})+\,\mbox{epi}(\bfu_{(s-1)\wedge\tau})= 
\,\mbox{epi}(v_{(s-1)\wedge\tau}).
\]
Furthermore, by the superhedging condition (\ref{3.17}),
\[
(\al_{(s-1)\wedge\tau},\,\be_{(s-1)\wedge\tau})\in\,\mbox{epi}(H_{s-1,\tau}),
\]
and so,
\begin{eqnarray*}
&(\al_{(s-1)\wedge\tau},\,\be_{(s-1)\wedge\tau})\in\,\mbox{epi}(\max(
H_{s-1,\tau},\, v_{(s-1)\wedge\tau}))\\
&\subset\,\mbox{epi}(\max(
q^b_{(s-1)\wedge\tau},v_{(s-1)\wedge\tau}))\subset\,\mbox{epi}(
u_{(s-1)\wedge\tau})
\end{eqnarray*}
since $H_{s-1,\tau}\geq q^b_{(s-1)\wedge\tau}$ and
$u_{(s-1)\wedge\tau}\leq\max(q^b_{(s-1)\wedge\tau},\, 
v_{(s-1)\wedge\tau})$ by (\ref{3.22}), completing the induction step. Now
taking $s=0$ in (\ref{4.6}) we obtain the assertion (d) proving the 
statement (ii).
\end{proof}

Now we are ready to prove Theorems \ref{thm3.1}-\ref{thm3.3}. In view of
Lemma \ref{lem4.1}(i) it follows from (\ref{2.7}) that
\begin{eqnarray}\label{4.7}
&V^a=\min\{ -\te_0(-\gam,-\del):\, (\gam,\del)\in\,\mbox{epi}(z_0)\}\\
&\leq\min\{ -\te_0(-\gam,0):\, (\gam,0)\in\,\mbox{epi}(z_0)\}\leq\min
\{\gam:\, (\gam,0)\in\,\mbox{epi}(z_0)\}=z_0(0).\nonumber
\end{eqnarray}

In order to derive the inequality in the other direction we will use
$z_{\sig,s},\, s=0,1,...,T$ constructed in Theorem \ref{thm3.2}(i) for
any $\sig\in\cT$. First, we show by the backward induction that for all
$s=0,1,...,T$,
\begin{equation}\label{4.8}
z_{\sig,s}\geq z_{\sig\wedge s}.
\end{equation}
Indeed, for $s=T$ we have $z_{\sig,T}=G_{\sig,T}\geq z_{\sig\wedge T}$
which is clear from the definitions when $\sig=T$ while if $\sig<T$ then
$z_{\sig,T}=q^a_\sig\geq z_\sig$. Suppose that (\ref{4.8}) holds true for
$s=T,T-1,...,t$. In order to prove it for $s=t-1$ we observe that
by (\ref{3.20}),
\[
z_{\sig,t-1}\geq G_{\sig,t-1}=q^a_{\sig}\geq z_\sig=z_{\sig\wedge(t-1)}
\]
on the set $\{\sig<t-1\}$. On the other hand, on the set $\{\sig\geq t-1\}$
by the definition and the induction hypothesis $\bf z_{\sig,t-1}\geq
\bf z_{t-1}$. This together with the monotonicity of the operator
gr$_{[d,c]}$ yields that $w_{\sig,t-1}\geq w_{t-1}$. Hence, on $\{\sig\geq
t-1\}$ by (\ref{3.20}),
\[
z^\mu_{t-1}\leq\max(r^a_{t-1}(\mu,x),\, w^\mu_{t-1}(x))\leq\max( G_{\sig,t-1}
(\mu,r),\, w^\mu_{\sig,t-1})=z^\mu_{\sig,t-1}
\]
completing the induction step, and so (\ref{4.8}) is valid for all 
$s=0,1,...,T$.

For each $\sig\in\cT$ set 
\begin{equation}\label{4.9}
V^a_\sig=\inf_{\pi}\{-\te_0(-\al_0,-\be_0):\, (\sig,\pi)\,\,\mbox{with}\,
\pi=(\al_t,\be_t)_{t=0}^T\,\mbox{being a superhedging strategy for the seller}\}.
\end{equation}
It is easy to see that $V^a_\sig$ is the seller's (ask) price of the
American option with the payoff process $Q_{\sig,t}=(Q_{\sig,t}^{(1)},
Q_{\sig,t}^{(2)}),\, t=0,1,...,T$ (see (\ref{2.4})) and the same transactions
costs setup as before. Hence, we can rely on Theorem 3.3 from \cite{RZ} in
order to claim that
\begin{eqnarray}\label{4.10}
&V^a_\sig=\max_{\chi\in\cX}\max_{(P,S)\in\bar\cP(\chi)}
\bbE_P(Q_{\sig,\cdot}^{(1)}+SQ_{\sig,\cdot}^{(2)})_\chi\\
&=\max_{\chi\in\cX}\sup_{(P,S)\in\cP(\chi)}
\bbE_P(Q_{\sig,\cdot}^{(1)}+SQ_{\sig,\cdot}^{(2)})_\chi=z_{\sig,0}(0).
\nonumber\end{eqnarray}
Combining (\ref{2.7}), (\ref{4.8}) and (\ref{4.10}) we obtain that
\begin{eqnarray}\label{4.11}
&V^a=\min_{\sig\in\cT}V^a_\sig=\min_{\sig\in\cT}\max_{\chi\in\cX}
\max_{(P,S)\in\bar\cP(\chi)}\bbE_P(Q_{\sig,\cdot}^{(1)}+
SQ_{\sig,\cdot}^{(2)})_\chi\\
&=\min_{\sig\in\cT}\max_{\chi\in\cX}\sup_{(P,S)\in\cP(\chi)}
\bbE_P(Q_{\sig,\cdot}^{(1)}+SQ_{\sig,\cdot}^{(2)})_\chi
=\min_{\sig\in\cT}z_{\sig,0}(0)\geq z_0(0).
\nonumber\end{eqnarray}
Now (\ref{4.7}) together with (\ref{4.11}) yields both (\ref{3.9}) from
Theorem \ref{thm3.1} and the conclusion of Theorem \ref{thm3.2}(i).

It follows from the equality $V^a=z_0(0)$ that $(V^a,0)\in\,$epi$(z_0)$.
On the other hand, we showed already in the proof of Lemma \ref{lem4.1}(i)
that any pair $(\sig,\pi)$ of $\sig\in\cT$ and $\pi=(\al_t,\be_t)_{t=0}^T\in
\Phi$ with $(\al_0,\be_0)\in\,$epi$(z_0)$ constructed by the algorithm of
Theorem \ref{thm3.3}(i) is a superhedging strategy for the seller. Hence,
the conclusion of Theorem \ref{thm3.3}(i) follows, completing the proof of
our results concerning the seller.  \qed

In order to obtain our results concerning the buyer we will rely on the
duality of the seller and the buyer positions in the setup of game options.
Indeed, since negative payoffs are also allowed we can view now the buyer
as a seller of the new game option with the payoff (2-vector) function
\[
\hat Q_{s,t}=(\hat Q_{s,t}^{(1)},\, \hat Q_{s,t}^{(2)})=-X_s\bbI_{s<t}-
Y_s\bbI_{t\leq s}=-Q_{s,t}
\]
while the former seller becomes a buyer of this new option. We observe 
the slight difference here that when $s=t$ the payoff $-Y_s\geq -X_s$
includes "the penalty" but this convention does not influence the results.
Now we see that replacing $Q^{(1)}_{s,\tau}$ and $Q^{(2)}_{s,\tau}$ by
$\hat Q^{(1)}_{s,\tau}$ and $\hat Q^{(2)}_{s,\tau}$ in the superhedging 
condition (\ref{2.8}) transforms it into the form (\ref{2.5}) and writing
(\ref{2.8}) for $\hat V^b=-V^b$ we transform (\ref{2.8}) into the form 
(\ref{2.7}). Next, $H_{s,t}$ in (\ref{3.15}) will have the form
of $G_{s,t}$ there if we replace  $Q_{s,t}^{(1)}$ and $Q_{s,t}^{(2)}$ 
there by $\hat Q_{s,t}^{(1)}$ and $\hat Q_{s,t}^{(2)}$, respectively.
Furthermore, (\ref{3.10}) will have the form of (\ref{3.9}) if we rewrite
it for $\hat V^b,\,\hat Q_{\cdot,\tau}^{(1)},\,\hat Q_{\cdot,\tau}^{(2)}$
in place of $V^b,\, Q_{\cdot,\tau}^{(1)},\, Q_{\cdot,\tau}^{(2)}$, 
respectively. Taking into account these arguments we derive (\ref{3.10})
and the assertions (ii) of Theorems \ref{thm3.2} and \ref{thm3.3} from
(\ref{3.9}) and from the corresponding assertions (i) of these theorems. \qed

\section{Shortfall risk and partial hedging}\label{sec5}

In this section we discuss the shortfall risk for game options with transaction
     costs. The shortfall risk of the seller using a cancellation stopping time
     $\sig$ and a self-financing strategy $\pi=(\al_t,\be_t)_{t=0}^T$ is
     defined by
     \begin{equation}\label{5.1}
     R^a(\sig,\pi)=\sup_{\tau\in\cT}\bbE_\bbP\big(\te_{\sig\wedge\tau}
     (\al_{\sig\wedge\tau}-Q^{(1)}_{\sig,\tau},\,
     \be_{\sig\wedge\tau}-Q^{(2)}_{\sig,\tau})\big)^-
     \end{equation}
     where $a^-=-\min(a,0)$. The shortfall risk of the seller wishing to
     spend no more than an amount $x$ in order to set his (partial) hedging
     portfolio is defined by
     \begin{equation}\label{5.2}
     R^a(x)=\inf_{\sig\in\cT,\,\pi}\{ R(\sig,\pi):\,\pi=(\al_t,\be_t)_{t=0}^T\,\,
     \mbox{is self-financing and}\,\, -\te_0(-\al_0,-\be_0)\leq x\}.
     \end{equation}
     
     Similarly, we can define the shortfall risk of the buyer using an
     exercise stopping time $\tau$ and a self-financing portfolio strategy
     $\pi=(\al_t,\be_t)_{t=0}^T$ by
     \begin{equation}\label{5.3}
     R^b(\tau,\pi)=\sup_{\sig\in\cT}\bbE_\bbP\big(\te_{\sig\wedge\tau}
     (\al_{\sig\wedge\tau}+Q^{(1)}_{\sig,\tau},\,
     \be_{\sig\wedge\tau}+Q^{(2)}_{\sig,\tau})\big)^-
     \end{equation}
     where $a^-=-\min(a,0)$. The shortfall risk of the buyer who takes
     a loan in the amount at least $x$ to pay it for the option and to
     set a (partial) hedging portfolio is defined by
     \begin{equation}\label{5.4}
     R^b(x)=\inf_{\tau\in\cT,\,\pi}\{ R^b(\tau,\pi):\,\pi=
     (\al_t,\be_t)_{t=0}^T\,\,\mbox{is self-financing and}\,\, 
     \te_0(-\al_0,-\be_0)\geq x\}.
     \end{equation}
     The shortfall risk of the buyer measures the maximal expected amount
     which may be needed to settle buyer's bank loan in addition to his 
     portfolio liquidation value and the option payoff he receives from 
     the seller.
     
     Observe that usually $V^a>V^b$ and in this case if the buyer agrees
     to pay for the option the upper hedging price $V^a$ (or any price
     higher than $V^b$) he should take into account the shortfall risk
     while the same is true concerning the seller if he agrees to sell
     the option for the lower hedging price (or any price lower than $V^a$).
     It is not difficult to see that in the study of the shortfall risk pure
     stopping times 
     suffice and there is no need to deal with randomized ones which would
     lead actually to the same values of the shorfall risk in view of linearity
     of corresponding expressions. Some dynamical programming algorithms for
     computation of shortfall risks and corresponding partial hedging
     strategies can be constructed for game options with proportional
     transaction costs in finite discrete markets similarly to \cite{Do1}
     where such algorithms were described for binomial models of American
     options but their extension to game options is also possible 
     (see Remark 6.3 in \cite{Do1}) which was done for frictionless 
     multinomial markets in \cite{DK}.

%\newpage

\end{document}